\documentclass[journal,draftcls,onecolumn,12pt,twoside]{IEEEtran}
\usepackage{algorithm} 
\usepackage{algpseudocode}
\usepackage{graphicx}
\usepackage{amssymb}
\usepackage{amsfonts}
\usepackage{amsmath}
\usepackage{float}
\usepackage{url}
\usepackage{bm}
\usepackage[bottom]{footmisc}
 \usepackage[colorlinks=true]{hyperref}
\hypersetup{urlcolor=blue, citecolor=blue}

\IEEEoverridecommandlockouts
\newtheorem{Theorem}{Theorem}
\newtheorem{Proposition}{Proposition}
\newtheorem{Corollary}{Corollary
}
\newtheorem{Definition}{Definition}
\newtheorem{Example}{Example}

\newtheorem{lemma}{Lemma}

\newenvironment{proof}{\textit{Proof\,:}} { $\blacksquare$}

\ifCLASSINFOpdf
\else
\fi
\begin{document}
\title{Edge Coloring Technique to Remove Small Elementary Trapping Sets from Tanner Graph of  QC-LDPC Codes  with Column Weight 4}
\author{Mohammad-Reza~Sadeghi and Farzane Amirzade\\
\thanks{%
Manuscript received  ??, ????; revised  ??, ????.}
\thanks{  M.-R. Sadeghi and F. Amirzade are with the Department of Mathematics and Computer Science, Amirkabir University of Technology.

(e-mail:  famirzade@gmail.com,  msadeghi@aut.ac.ir).}
 \thanks{%
 Digital Object Identifier ????/TCOMM.?????}}


\maketitle
\begin{abstract}
One of the phenomena that  causes high decoding failure rates is trapping sets. Characterization of $(a,b)$ elementary trapping sets (ETSs), their graphical properties and the lower bounds on their size in variable regular LDPC codes with column weights 3, 4, 5 and 6, where $a$ is the size of the ETS and $b$ is the number of degree-one check nodes,   have been an interesting subject among researchers. Although progressive-edge-growth method (PEG) to construct LDPC codes free of an specific ETS has been proposed in the literature, it is mostly applied to  LDPC codes with column weight 3. In this paper, we focus on constructing QC-LDPC codes with column weight 4 whose Tanner graphs are free of small ETSs. Using  coloring the edges of the variable node (VN) graph corresponding to an ETS, we provide the sufficient conditions to obtain QC-LDPC codes with column weight 4, girth 6 and free of $(5,b)$ ETSs, where $b\leq4$, and $(6,b)$ ETs, where $b\leq2$. Moreover, for $(4,n)$-regular QC-LDPC codes with girth 8, we present a method to remove $(7,4)$ ETSs from Tanner graphs. 
  \end{abstract}

%
\IEEEpeerreviewmaketitle
\section{Introduction}
\IEEEPARstart{Q} uasi-cyclic low-density parity-check (QC-LDPC) codes are an essential  category of LDPC codes. One of the approaches for constructing QC-LDPC codes is graph-theoretic-based whose most well-known methods are progressive edge growth (PEG) and  protograph-based methods. The most important representation of LDPC codes is Tanner graph.  The length of the shortest cycles of the Tanner graph, girth, has been known to influence the code performance. An induced subgraph of the Tanner graph on $a$ variable nodes and $b$ check nodes of odd degrees is called an $(a,b)$ trapping set of size $a$. Empirical results have shown that, in addition to short cycles, trapping sets influence the performance of  LDPC codes.  Among them, the most harmful ones which cause high decoding failure rates and exert a strong influence on the error floor are those  with check nodes of degree 1 or 2 which is called elementary trapping sets (or simply ETSs).

Characterizations of ETSs in variable-regular LDPC codes with  column weight values of $3,4,5,6$, girths $6,8$ and $a,b\leq10$ were presented in \cite{2010}-\cite{2016}.  The tightest lower bounds on the size of ETSs in variable-regular LDPC codes with different girths were provided in $\cite{farzane}$ which demonstrate that as the girth increases, the lower bound on the size of trapping sets will increase too. Besides characterizing trapping sets, avoiding small trapping sets to reduce the error floor of LDPC codes is a challenging problem. Raising the girth is known as a technique to avoid harmful trapping sets with small sizes. Another method to remove ETSs with small size is progressive-edge-growth (PEG). Using PEG algorithm, QC-LDPC codes whose permutation matrices are obtained from Latin Squares and their Tanner graphs are free of some small trapping sets were constructed in $\cite{Vasic2}$. An improved version of  PEG algorithm in $\cite{Vasic}$ results in an  LDPC code with column weight 3 whose Tanner graph has girth 8, is free of $(5,3)$ trapping sets and contains a minimum number of $(6,4)$ trapping sets. Avoiding some 8-cycles in the Tanner graph of a fully connected  $(3,n)$-regular QC-LDPC code  with girth 8 causes the removal of $(5,3)$ and $(7,3)$ ETSs  which result in a Tanner graph free of $(a,b)$ ETSs, where $a\leq8$ and $b\leq3$, \cite{main}. Edge coloring technique to determine  occurrence and avoidance of ETSs in the Tanner graph of QC-LDPC codes with different column weights was first introduced in \cite{TCOMM2}.  For example, it was proved that Tanner graphs of fully connected $(3,n)$-regular QC-LDPC codes with girth 6 are free of the $(a,1)$ ETS. In addition, sufficient conditions for the exponent matrices  to obtain  QC-LDPC codes with column weight 3, minimum lifting degree, whose Tanner graphs have girths 6 and 8  and  are free of   small size ETSs were presented.  

As mentioned above, although graphical structures of trapping sets in LDPC codes with different column weights have been presented in the literature, most of constructed LDPC codes whose Tanner graphs are free of ETSs with small size have a column weight 3.  In this paper, for the first time, sufficient conditions for an exponent matrix to construct  $(4,n)$-regular QC-LDPC codes whose Tanner graphs have girths 6 and 8 and are free of small ETSs are provided. In fact, we first propose some conditions which make coloring the edges of the VN graph, corresponding to an specific ETS, with four colors impossible. Since in fully connected QC-LDPC codes the non-existence of 4-edge-coloring in a VN graph is equivalent to the non-existence of its corresponding ETS in Tanner graph, \cite{TCOMM2}, we conclude that our proposed conditions cause the removal of those specific ETSs from Tanner graph. 

The smallest size of an ETS in Tanner graph of LDPC codes with column weight 4 is 5, \cite{farzane}. For $(4,n)$-regular QC-LDPC codes with girth 6, we prove that the removal of 6-cycles obtained from the first three rows of the exponent matrix and the removal of 6-cycles obtained from rows with row indices 1, 2 and 4 cause the removal of $(5,b)$ ETSs, where $b\leq4$, and $(6,b)$ ETs, where $b\leq2$. For $(4,n)$-regular QC-LDPC codes with girth 8, we first prove that the number of non-isomorphic $(7,4)$ ETSs, which is the smallest ETS in girth 8 LDPC codes, \cite{farzane}, is one only. We, then, prove that avoiding 8-cycles obtained from three rows of the exponent matrix including the first row twice and the second row at least once prevents a 4-edge-coloring for the VN graph of $(7,4)$ ETS, which is a complete bipartite graph $K_{3,4}$. Therefore, avoiding these 8-cycles provides us with $(4,n)$-regular QC-LDPC codes  whose Tanner graphs have girth 8 and are free of $(7,4)$ ETSs.  In order to reduce the size of the search space we use the numerical results in the literature for QC-LDPC codes with column weight 3.
     
\raggedbottom The rest of the paper is organized as follows. Section \ref{II} presents some basic notations, definitions and our main tools such as edge coloring VN graphs corresponding to the ETSs to prove our results.  In Sections \ref{IV} and \ref{V}, respectively, we consider sufficient conditions to have QC-LDPC codes with girths 6 and 8 whose Tanner graphs are free of some small ETSs. In the last section we summarize our results. 
\section{Preliminaries}\label{II}
Let $N$ be an integer number. Consider the following exponent matrix $B=[b_{ij}]$, where $b_{ij}\in \lbrace 0,1,\cdots,N-1\rbrace$, 
\begingroup\fontsize{11pt}{11pt}\begin{align}\label{Rela2}
B=\left[\begin{array}{cccc}
b_{00}&b_{01}&\cdots &b_{0(n-1)}\\
b_{10}&b_{11}&\cdots &b_{1(n-1)}\\
\vdots &\vdots &\ddots &\vdots \\
b_{(m-1)0}&b_{(m-1)1}&\cdots &b_{(m-1)(n-1)}\\
\end{array}\right].
\end{align}\endgroup

To construct the parity-check matrix of a QC-LDPC code, all $\infty$ elements of the matrix, $B$, is substituted by an $N\times N$ zero matrix and the $ij$-th element of  $B$, which is a positive integer number, is substituted by an $N\times N$ matrix $I^{b_{ij}}$. This matrix is a circulant permutation matrix (CPM) in which the single 1-component of the top row is located at the $b_{ij}$-th position and other entries of the top row are zero. The $r$-th row of the circulant permutation matrix is formed by  $r$ right cyclic shifts of the first row and clearly the first row is a right cyclic shift of the last row.  The null space of the parity-check matrix gives us a QC-LDPC code.

The necessary and sufficient condition for the existence of cycles of the length $2k$  in  the Tanner graph of QC-LDPC codes was provided in $\cite{2004}$. This well-known result is our principle tool and we summarize it as follows. If
\begingroup\fontsize{13pt}{13pt}\begin{align}\label{Equation}
\sum_{i=0}^{k-1}(b_{m_in_i}-b_{m_in_{i+1}})=0  \mod N,
\end{align}\endgroup
where $n_k=n_0,\ m_i\neq m_{i+1},\ n_i\neq n_{i+1}$ and $b_{m_in_i}$ is the $(m_i,n_i)$-th entry of $B$, then the Tanner graph of the parity-check matrix has cycles of the length $2k$. Equation (\ref{Equation}) proves that the cycle distribution of a code is fully described by its exponent matrix and lifting degree. If we aim to find a QC-LDPC code with girth, $g$, and the lifting degree, $N$, from the exponent matrix, $B$, we have to choose the elements $b_{ij}s$ belong to the set $\{0,\dots,N-1\}$
	such that Equation (\ref{Equation}) are avoided for values of $k < \frac{g}{2}$.

Tanner graph is a bipartite graph. Suppose the set of variable nodes, $V$, forms one set of the Tanner graph and the set of check nodes, $C$, forms another set of the Tanner graph.
\begin{Definition}
An $(a,b)$ trapping set is an induced subgraph of the Tanner graph on a subset $S$ of $V$, where  $|S|=a$, and contains $b$ check nodes of odd degrees,  referred to as  unsatisfied check nodes. The other check nodes have even degrees named as satisfied check nodes. An $(a,b)$ trapping set is called elementary if all check nodes are of degree 1 or 2. As a result, all unsatisfied check nodes in an elementary trapping set (or ETS) have degree one. 
\end{Definition}

For a bipartite graph $G$ corresponding to
an ETS, a $variable\ node$ (VN) $graph$ is
constructed by removing all degree-one check nodes, defining variable nodes of G as its vertices and degree-two check nodes connecting the variable nodes in $G$ as
its edges. 

Any cycle of length $k$ in the VN graph is equivalent to a $2k$-cycle in its corresponding ETS. For example, a sequence of $v_0,c_0,v_1,c_1,v_2,c_2$ as a 6-cycle in an ETS, where $v_i\in V$ and $c_i\in C$, is equivalent to a triangle with vertices $v_1,v_2,v_3$ in the VN graph.  Therefore, in a Tanner graph with girth at least 8, the VN graph of each ETS is a triangle-free graph. 

In the following we provide some well-known theorems including Vizing's theorem in graph theory \cite{Bondy} as well as the consequences related to  coloring the edges of the VN graphs in \cite{TCOMM2},  which are our principle tools to achieve our results. 

\begin{Definition}\label{Edge-coloring}
	An edge coloring of a graph, $G$, is an assignment of colors (labels)  to the edges of the graph so that no two adjacent edges have the same color (label). The minimum required number of colors for the edges of a given graph is called the chromatic index of the graph which is denoted by ${\mathcal{X'}}(G)$, or simply  ${\mathcal{X'}}$, which indicates the graph has ${\mathcal{X'}}$-edge-coloring.
\end{Definition}

\begin{Theorem}\label{Vising}
\cite{Bondy}, if $\Delta(G)$ is the maximum degree of a graph $G$, then $$\Delta(G)\leq {\mathcal{X'}}\leq \Delta(G)+1.$$
\end{Theorem}

A close connection between the chromatic index of a VN graph and the occurrence  of its corresponding ETS in Tanner graph of fully connected QC-LDPC codes has been proposed in \cite{TCOMM2} which is as follows.
\begin{Proposition}\label{Prop1}
	Given a fully connected $(m,n)$-regular QC-LDPC code. The necessary condition for the Tanner graph to contain any $(a,b)$ ETS  is that the corresponding VN graph has $m$-edge-coloring.
\end{Proposition} 

We note that all of our results are applicable  to fully connected QC-LDPC codes whose exponent matrices are free of $\infty$ elements. For example,  LDPC codes with column weight $2\ell$ and girth 6 include $(2\ell+1,0)$, $(2\ell+1,2)$ ETSs, whereas Example \ref{EX4} and Corollary \ref{EX5}, proved in \cite{TCOMM2}, show that fully connected $(2\ell,n)$-regular QC-LDPC codes with girth 6 contain no	$(2\ell+1,0)$, $(2\ell+1,2)$ ETSs.  In order to clarify these consequences we also need other definitions in graph theory such as complete graph and independent sets.
\begin{Definition}
	A complete  graph is a graph in which every pair of distinct vertices are connected by a unique edge. A complete  graph on $n$ vertices is denoted by $K_{n}$.   
\end{Definition}

\begin{Definition}
	An independent edge set, also called matching, of a graph, $G$, is a subset of the edges such that no two edges in the subset share a vertex in. An independent edge set with the maximum cardinality is called a maximum independent edge set whose cardinality is denoted by $\alpha'(G)$.
\end{Definition}
\begin{Example}\label{EX4}
	A complete graph on $2\ell+1$ vertices is equivalent to the VN graph corresponding to a $(2\ell+1,0)$ ETS with girth 6. It is proved that the chromatic index of a complete graph on $2\ell+1$ vertices is  ${\mathcal{X'}}(K_{2\ell+1})=2\ell+1$.  Therefore, Proposition \ref{Prop1} proves that the Tanner graph of a fully connected $(2\ell,n)$-regular QC-LDPC code with girth 6 contains no $(2\ell+1,0)$ ETS, because its VN graph has no $2\ell$-edge-coloring.
\end{Example}

\begin{Theorem}\label{green}
	Suppose $E(G)$ is the set of edges in a graph $G$. If $|E(G)|>\alpha'(G)\Delta(G)$, then ${\mathcal{X'}}=\Delta(G)+1$, \cite{Green}. 
\end{Theorem}

\begin{Corollary}\label{EX5}
	A fully connected $(2\ell,n)$-regular QC-LDPC code with girth 6 contains no  $(2\ell+1,2)$ ETS.
\end{Corollary}

A straight consequence of the above discussions is as follows. Fully connected $(4,n)$-regular QC-LDPC codes with girth 6 contain no	$(5,0)$, $(5,2)$ ETSs.

\section{QC-LDPC Codes With Girth 6, Column Weight 4 and Free of Small ETSs}\label{IV}

The smallest size of $(a, b)$ ETSs in Tanner graph of a $(4,n)$-regular QC-LDPC code with
girth 6 is $a = 5$, \cite{farzane}. In this section using Vizing’s theorem in graph theory, \cite{Bondy}, we consider
sufficient conditions to remove ETSs with small size. In fact, our goal is to construct $(4,n)$-
regular QC-LDPC code with girth 6 whose Tanner graphs are free of $(5,b)$ ETSs, where $b\leq4$,
and $(6,b)$ ETSs, where $b\leq2$. In \cite{TCOMM2}, we proved that Tanner graph of fully connected  $(4,n)$-
regular QC-LDPC code with girth 6 is free of $(5,0)$ and $(5, 2)$ ETSs. Therefore, to reach our goal
we focus on fully connected QC-LDPC codes and provide the sufficient condition to remove
three ETSs with parameters $(5,4),\ (6, 0)$ and $(6,2)$. Hereafter, 6-cycles obtained from three rows with
row indices $i, j, k$, in an exponent matrix is denoted by 6-cycle$_{\{i,j,k\}}$.
\begin{lemma}\label{lemma1}
	Suppose $B$ is the exponent matrix of a fully connected $(4,n)$-regular QC-LDPC
	code. The sufficient condition to remove $(5,4)$ ETSs in Tanner graph is to avoid 6-cycle$_{\{1,2,3\}}$ and 6-cycle$_{\{1,2,4\}}$.

\end{lemma}
\begin{proof}
There are two non-isomorphic $(5,4)$ ETSs. We call them type 1 and type 2 . Suppose row indices of the exponent matrix are 1, 2, 3, and 4. According
to Proposition \ref{Prop1}, since the edges of VN graphs corresponding to these ETSs can be colored by
these four labels, Tanner graph contains these ETSs. Our goal is to construct fully connected
$(4,n)$-regular QC-LDPC codes with girth 6 whose Tanner graphs are free of $(5,4)$ ETSs.

Each triangle in the VN graph is equivalent to a 6-cycle in Tanner graph. So, avoiding 6-cycle$_{\{1,2,3\}}$ is equivalent to coloring edges of VN graphs such that no
triangles is colored by three labels 1, 2 and 3. If we color edges of the VN graph corresponding to type 1, then by avoiding triangles with labels 1, 2 and 3 we end up with a variable node with two adjacent edges with the same color. Different ways to color edges of the VN graph type 1 can be illustrated in which a variable-node $v$ is connected to four variable nodes $v_1, v_2, v_3$ and $v_4$. None of the coloring methods results in a 4-edge coloring. Since, a 4-edge coloring with the mentioned restriction is impossible, Proposition \ref{Prop1} proves that avoiding 6-cycle$_{\{1,2,3\}}$ results in Tanner graphs which are free of the $(5,4)$ ETS, type 1. A 4-edge-coloring for the VN graph of type 2 can be obtained which contains no triangles with labels 1, 2 and 3. Therefore, in order to remove this kind of ETSs we also avoid 6-cycle$_{\{1,2,4\}}$. Suppose edges of a triangle in the VN graph of an ETS is colored by one of the following triples, $(1, 2, 3),\ (1, 2, 4),\ (1, 3, 4),\ (2, 3, 4)$.
Avoiding 6-cycle$_{\{1,2,3\}}$ and 6-cycle$_{\{1,2,4\}}$, is equivalent to coloring the edges of the VN graph in which triangles are not colored with triples $(1, 2, 3)$ or $(1, 2, 4)$. Since, coloring edge of the VN graph of type 2 with 4 colors whose triangles are colored with triples $(1, 3, 4),\ (2, 3, 4)$ is impossible, Proposition \ref{Prop1} proves that, with this limitation all of ETSs with a VN graph of type 2 are removed from Tanner graph.
\end{proof}

\begin{lemma}\label{lemma2}
	If Tanner graph of a $(4, n)$-regular QC-LDPC code with girth 6 contains no $(5, 4)$
	ETSs, then it contains no $(6, 0)$ ETSs too.
\end{lemma}
\begin{proof}
	Removing each variable node from the VN graph of a $(6, 0)$ ETS results in the VN
	graph corresponding to a $(5, 4)$ ETS. Therefore, removing $(5, 4)$ ETSs causes the removal of
	$(6, 0)$ ETSs.
\end{proof}
\begin{lemma}\label{lemma3}
	Suppose $B$ is the exponent matrix of a fully connected $(4, n)$-regular QC-LDPC
	code. The sufficient condition to remove $(6, 2)$ ETSs in Tanner graph is to avoid 6-cycle$_{\{1,2,3\}}$ and 6-cycle$_{\{1,2,4\}}$.
\end{lemma}
\begin{proof}
	Similar to the proof of Lemma \ref{lemma1}, by avoiding 6-cycles from the first three rows and 6-cycles obtained from the first, second and last rows of $B$, edges of a triangle in the VN graph	of an ETS can be colored by a triple $(1, 3, 4)$ or $(2, 3, 4)$. There are three non-isomorphic $(6, 2)$ ETSs. The VN graphs of two of them which contain a complete graph $K_4$ have the $(5, 4)$ ETS with a VN graph of type 2 as their subgraph. Since a 4-edge-coloring for the
	VN graph, type 2, without triangles with colors $(1, 2, 3)$ and $(1, 2, 4)$ is impossible, there is no	4-edge coloring for the VN graph of a $(6, 2)$ ETS by the mentioned restriction. 
	
	If the VN graph contains no $K_4$, then contains a variable node, $v$, of degree 4 which is connected to four variable nodes $v_1, v_2, v_3$ and $v_4$. These five variable nodes contain four triangles. We color edges connected to $v$ such that edges with colors 1 and 2 are not in a triangle, because $(1, 2, 3)$ or $(1, 2, 4)$ are abandoned. By this assumption and continuing the process of coloring the edges we end up with a variable node with two adjacent edges with the same color.
	
	Generally, the removal of 6-cycle$_{\{1,2,3\}}$ and 6-cycle$_{\{1,2,4\}}$ causes the non-existence of 4-edge coloring for the VN graph of $(6, 2)$ ETSs. Consequently, Proposition 1, proves Tanner graph is free of $(6, 2)$ ETSs.
\end{proof}

\begin{Theorem}
	The sufficient condition to have $(4, n)$-regular QC-LDPC codes with girth 6 whose
	Tanner graphs are free of $(5, b)$ ETSs, where $b\leq4$, and $(6, b)$ ETSs, where $b\leq2$ is to construct fully connected exponent matrices in which 6-cycle$_{\{1,2,3\}}$ and 6-cycle$_{\{1,2,4\}}$ are avoided.
\end{Theorem}
\begin{proof}
	The three Lemmas \ref{lemma1}, \ref{lemma2} and \ref{lemma3} result in the fully connected $(4, n)$-regular QC-LDPC codes with the mentioned property.
\end{proof}

In Table \ref{Tabel}, we provide the exponent matrices to construct fully connected $(4, n)$-regular QCLDPC
codes with girth 6 whose Tanner graphs are free of $(5, b)$ ETSs, where $b\leq4$, and $(6, b)$
ETSs, where $b\leq2$.

\section{QC-LDPC Codes With Girth 6, Column Weight 4 and Free of Small ETSs}\label{V}
In \cite{farzane} we proved that the smallest elementary trapping set in a $(4, n)$-regular QC-LDPC code
with girth 8 is a $(7, 4)$ ETS and we took a complete bipartite graph $K_{3,4}$ as the VN graph of this ETS. In this section, we first prove the number of non-isomorphic $(7, 4)$ ETSs is one only. Then we propose a method to remove them from Tanner graphs.
\begin{Proposition}\label{Prop2}
	The VN graph of a $(7, 4)$ ETS is a bipartite graph.
\end{Proposition}
\begin{proof}
	A graph, $G$, is a bipartite graph if and only if it has no odd cycle. Therefore, it is
	sufficient to prove the VN graph, $G$, is free of odd cycles. Since Tanner graph has girth 8, every ETS is 6-cycle free which results in a triangle-free VN graph.
	
	Since a $(7, 4)$ ETS has four 1-degree check-nodes, at least three variable nodes of $G$ have degree 4. Let the VN graph have a 5-cycle. To draw such VN graph we first depict a 5-cycle, $C_5$. Connecting two nonadjacent vertices on $C_5$ causes a triangle. Therefore, there is no connection between nonadjacent vertices on $C_5$. There are two variable nodes outside of this cycle. To avoid triangles, each of these two variable nodes can be connected to at most two vertices on $C_5$. Therefore, the maximum number of variable nodes with degree 4 is at most 2, which is a contradiction. Therefore the VN graph is 5-cycle free.
	
	Suppose the VN graph has a 7-cycle. In order to depict $G$ we first draw a 7-cycle, $C_7$. Let three variable nodes $u, v_1, v_2$ be on $C_7$, $u$ has degree 4 and $v_1, v_2$ are the neighbors of $u$ such that $uv_1$ and $uv_2$ are not edges of the cycle. The existence of these two adjacent edges $uv_1$ and $uv_2$ causes a triangle which is impossible. Therefore the VN graph is 7-cycle free.
\end{proof}
\begin{Proposition}\label{Prop3}
	The number of non-isomorphic $(7, 4)$ ETSs is one.
\end{Proposition}
\begin{proof}
As we mentioned in Proposition 2\ref{Prop2}, the VN graph of a $(7, 4)$ ETS, denoted by $G$, is a bipartite graph. There are three types of degree sequences for the variable nodes of $G$.
If an ETS contains four check nodes connected to four variable nodes, then $G$ has the degree
sequence $d_1 = \{4,4,4,3,3,3,3\}$. In this case, $G$ is a complete bipartite graph $K_{3,4}$.

If an ETS has two check nodes connected to a variable node and two check nodes connected to
another variable node, then $G$ has the degree sequence $d_2 = \{4,4,4,4,4,2,2\}$. Since $G$ is a bipartite graph, the existence of a vertex of degree 4 requires a part to have at least four variable nodes. The maximum number of vertices in the other part is three. Therefore, the existence of five variable nodes of degree 4 is impossible.

 The third type of degree sequence happens when there are two check nodes connected to a variable node and two check nodes connected to two another disjoint variable nodes. In this
 case the degree sequence is $d_3 = \{4,4,4,3,3,2,2\}$. Since there are three vertices of degree 4 we put four variable nodes in one part and the other three nodes in the other part. In this case we cannot have the degree sequence $d_3$. Generally, the only possible case is d1 for which the variable node is $K_{3,4}$.
\end{proof}

The VN graph of a $(7, 4)$ ETS is 4-edge coloring. To provide a method to remove this ETS
from Tanner graph we propose some restriction in coloring the edge of the VN graph with four
labels $1, 2, 3, 4$. In fact, we provide conditions under which the VN graph is not 4-edge coloring.
\begin{Theorem}
	The sufficient condition to have Tanner graphs with girth 8 and free of $(7, 4)$
	ETSs is to avoid 8-cycles obtained from three rows in which the first row is used twice and the second row is used at least once.
\end{Theorem}
\begin{proof}
	According to Proposition \ref{Prop3}, there is only one $(7, 4)$ ETS whose VN graph is a complete bipartite graph $K_{3,4}$. Suppose variable nodes in a part, $U$, of $K_{3,4}$ are denoted by $u_1, u_2, u_3$, and the ones in the other part, $V$ , are $v_1, v_2, v_3, v_4$. Without loss of generality, let the edges
	$u_1v_1, u_1v_2, u_1v_3, u_1v_4$ be colored by labels 1, 2, 3, 4, respectively. Our goal is to color the edges of $K_{3,4}$ such that no 4-cycle contains the label 1 twice and the label 2 at least once. In fact, no 4-cycle is colored by quadruples $(1,2,1,2),\ (1,2,1,3),\ (1,2,1,4)$, we denote these quadruples
	by a set $Q = \{(1,2,1,2),\ (1,2,1,3),\ (1,2,1,4)\}$. Since $K_{3,4}$ is complete, there is a 4-cycle between each two vertices in part $U$ and each two vertices in part $V$ . There are three options to color an edge $u_2v_1$. In the following, we elaborate the problems which we encounter with for each label of the edge $u_2v_1$.
	
	If we choose 2 for $u_2v_1$, then we have to choose 1 for edges $u_2v_2$ or $u_2v_3$ or $u_2v_4$. However, by this assumptions, we have one of the quadruples in $Q$ for 4-cycles $u_1v_1u_2v_2, u_1v_1u_2v_3, u_1v_1u_2v_4$.
	Hence, to color the edge u2v1 we cannot select the label 2. The same discussion holds for the edge $u_3v_1$.
	
	There are three options to color the edge $u_2v_2$ which are 1, 3, 4. If we choose 1 for $u_2v_2$, then we have a 4-cycle, $v_1u_1v_2u_2$ with the quadruple $(1,2,1,2)$ which contradicts our assumption. Therefore, there are two labels 3 and 4 for the edge $u_2v_2$. We have the same scenario for the
	edge $u_3v_2$. Since possible labels for the edge $u_2v_1$ are also 3 and 4, possible labels for two edges $u_2v_3$ and $u_2v_4$ are 1 and 2. Moreover, since there are two labels 3 and 4 for the edges
	$u_3v_1$ and $u_3v_2$, the labels for $u_3v_3$ and $u_3v_4$ can be 1 and 2. Consequently, we encounter with a 4-cycle whose edges have two options 1 and 2, only, which results in the quarter $(1,2,1,2)$, a contradiction. From these discussions we conclude that a 4-edge coloring for the edges of $K_{3,4}$ on the condition that none of quadruples in the set $Q$
	appears in a 4-cycle is impossible. Therefore, the VN graph of the $(7, 4)$ ETS with the mentioned
	restriction is not 4-edge coloring. According to Proposition \ref{Prop1}, the removal of 8-cycles obtained
	from three rows in which the first row is used twice and the second row is used at least once
	causes a $(4, n)$-regular QC-LDPC code with girth 8 and free of $(7, 4)$ ETSs.
\end{proof}

\begin{table}[h]
	\begin{center}
		\caption{Fully connected $(4,n)$-regular QC-LDPC codes with $n=5,6,7$ whose Tanner graphs have girth 6 and are free of $(5,b)$ ETSs, where $b\leq4$, as well as $(6,b)$ ETSs,  where $b\leq2$, OR Tanner graphs have girth 8 and   are free of $(7,4)$ ETSs}\label{Tabel}
		\begin{tabular}{|c|c|c|c|c|c|}
			\hline
			$n$&$N\ for\ girth\ 6$&$ exponent\ matrices$&$N\ for\ girth\ 8$&$ exponent\ matrices$\\
			\hline
			5& 13 & $\begin{array}{cccc}
			1 & 3 & 7 & 11 \\
			4& 12& 2& 5  \\
			10& 4& 5& 6\\
			\end{array}$&41&  $\begin{array}{cccc}
			1& 4& 11& 29\\
			2& 8 &17& 22\\
			14& 35& 33& 9\\
			\end{array}$\\
			\hline
			6&18& $\begin{array}{ccccc}
			4& 7& 8& 15& 17\\
			9& 11& 6& 10& 14\\
			13& 14& 2& 5& 3\\
			\end{array}$&63&  $\begin{array}{ccccc}
			1& 13& 16& 33& 39\\
			2& 7& 11& 21& 29\\
			4& 58& 22& 56& 14\\
			\end{array}$\\
			\hline
			7&21 & $\begin{array}{cccccc}
			2& 3& 4& 9& 14& 17\\
			10& 6& 15& 18& 19& 16\\
			13& 18& 10& 12& 16& 1\\
			\end{array}$&91&  $\begin{array}{cccccc}
			1& 4& 13& 30& 40& 45\\
			2& 8& 22& 33& 56& 75\\
			14& 48& 67& 85& 25& 83\\
			\end{array}$\\
			\hline
		\end{tabular}
	\end{center}
\end{table}

\section{Conclusion}\label{VI}
 The concept of edge coloring to determine the occurrence and avoidance of elementary trapping sets (ETSs) in the Tanner graph of fully connected regular QC-LDPC codes has been proposed recently in the literature. In this paper, we extended this topic to  QC-LDPC codes with column weight 4 which results in sufficient conditions for exponent matrices to have $(4,n)$-regular QC-LDPC codes with girths 6 and 8 whose Tanner graphs are free of small ETSs. We proved that in a QC-LDPC code with column weight 4 and girth 6 avoiding 6-cycles obtained from the first three rows of the exponent matrix and avoiding 6-cycles obtained from the first, second and the last rows of the exponent matrix cause the removal of $(5,b)$ ETSs, where $b\leq4$, as well as $(6,b)$ ETSs,  where $b\leq2$. Moreover, for girth 8  Tanner graphs whose variable nodes have degree 4 the sufficient condition to remove $(7,4)$ ETSs from Tanner graph is avoiding 8-cycles obtained from three rows of the exponent matrix in which the first row is used twice and the second row is used at least once.
 

\end{document}